\documentclass[conference, 9pt]{IEEEtran}

\usepackage{amsmath,amssymb,amsfonts}
\usepackage{algorithmic}
\usepackage{graphicx}
\usepackage{siunitx}
\usepackage{textcomp}
\usepackage{xcolor}
\usepackage{enumerate}
\usepackage{hyperref}
\usepackage{amsthm}
\hypersetup{
  colorlinks   = true, 
  urlcolor     = blue, 
  linkcolor    = blue, 
  citecolor   = red 
}
\usepackage{tikz}
\def\BibTeX{{\rm B\kern-.05em{\sc i\kern-.025em b}\kern-.08em
    T\kern-.1667em\lower.7ex\hbox{E}\kern-.125emX}}
    \usetikzlibrary{decorations.pathreplacing}

\usepackage{comment}
\usepackage[textsize=tiny]{todonotes}

\usepackage[style=ieee, maxnames=5, minnames=1, date=year, doi=false,isbn=false,url=false,backend=biber, sortcites, dashed=false]{biblatex}
 \definecolor{Red}{HTML}{E53E30}  
\definecolor{Green}{HTML}{00AD69}  
\definecolor{Blue}{HTML}{2171b5}
\definecolor{Purple}{HTML}{652F6C}  

\newtheorem{theorem}{Theorem}
\newtheorem{lemma}{Lemma}
\newtheorem{proposition}{Proposition}

\begin{document}

\title{Checking errors in reversible circuits}
\title{Reversible Circuits are Different: \\{\Large The Implications of Reversibility on Checking Errors}}
\title{Characteristics of Reversible Circuits for Verification:\\{\LargeThe Implications of Reversibility on Checking Errors}}
\title{Characteristics of Reversible Circuits for Error Detection}

\author{Lukas Burgholzer\IEEEauthorrefmark{1}
\hspace{5em} Robert Wille\IEEEauthorrefmark{1}\IEEEauthorrefmark{2}
\hspace{5em} Richard Kueng\IEEEauthorrefmark{1}\\
\IEEEauthorrefmark{1}Institute for Integrated Circuits, Johannes Kepler University Linz, Austria\\
\IEEEauthorrefmark{2}Software Competence Center Hagenberg GmbH (SCCH), Austria\\
\{lukas.burgholzer, robert.wille, richard.kueng\}@jku.at\\
\url{https://iic.jku.at/eda/research/quantum/}
}

\maketitle

\begin{abstract}


In this work, we consider error detection via simulation for reversible circuit architectures.
We rigorously prove that reversibility augments the performance of this simple error detection protocol to a considerable degree. 
A single randomly generated input is guaranteed to unveil a single error with a probability that only depends on the size of the error, not the size of the circuit itself. Empirical studies confirm that this behavior typically extends to multiple errors as well.  
In conclusion, reversible circuits offer characteristics that reduce masking effects -- a desirable feature that is in stark contrast to irreversible circuit architectures.


\end{abstract}


\section{Introduction}

The detection of errors is a fundamental problem in electrical engineering and computer science.
Given two circuits~$C_1$ and $C_2$ with $n$ inputs and $m$ outputs the task is to decide whether they describe the same functionality on the logical level. 

Many approaches exist that address this important and challenging problem. 
In this work, we focus on error detection protocols that only require simulation runs of the two circuits---as opposed to formal verification techniques which explicitly utilize structural knowledge about both circuits~\cite{ dischCombinationalEquivalenceChecking2007, marques-silvaCombinationalEquivalenceChecking1999, molitorEquivalenceCheckingDigital2010, jhaEquivalenceCheckingUsing1997, clarkeModelChecking2018, biereSymbolicModelChecking1999}. 
This is a severe restriction, but simulations alone are---in principle---sufficient to solve this task. If the two circuits are equivalent, they have the same input-output behavior.
Conversely, suppose that they are functionally distinct. Then, there exists at least one input string for which the two circuits produce \emph{distinct} outputs. In formulas:
\begin{equation}
\exists \vec{x}
\in \left\{0,1\right\}^n \quad \text{such that} \quad 
C_1(\vec{x}) \neq C_2 (\vec{x}).
\label{eq:witness}
\end{equation}
Such an input successfully detects the discrepancy between $C_1$ and $C_2$ and serves as a counterexample for the equivalence of both circuits.

The problem, however, is how to find counterexamples \eqref{eq:witness}.
If we only allow simulations of both circuits, i.e., we consider them as black boxes,
we do not have actionable advice on how to choose promising input strings and we may as well generate inputs uniformly at random: \mbox{$\vec{x} \sim \text{Unif}\left( \left\{0,1\right\}^n\right)$}, i.e., we flip an unbiased coin for each input value ($\vec{x}=(x_n,\ldots,x_1)$, where $x_n,\ldots,x_1 \sim x$ and $\mathrm{Pr} \left[x=0\right]=\mathrm{Pr}[x=1]=1/2$).
Subsequently, we simulate both circuits with this input and check whether they produce the same output: $C_1 (\vec{x})\overset{?}{=}C_2 (\vec{x})$.
If the outputs are distinct, we have found a counterexample.
The circuits cannot be equivalent.
But if the outputs are the same, the test is inconclusive. In this case, we must repeat it with new (randomly generated) inputs until we either find a counterexample (non-equivalence) or have exhausted all $2^n$ possible inputs (equivalence).  
The latter, unfortunately, can be a very real possibility. The two circuits $C_1$ and $C_2$ may differ on a single input only and it is extremely unlikely to quickly find this input by (random) chance. 


To make matters worse, classical circuits can mask even ``small'' errors very effectively. For $n=8$, this is illustrated in Fig.~\ref{fig:irreversible}. A cascade of logical \textsc{AND} gates, realizing the functionality \mbox{$y=x_n\cdot\hdots\cdot x_1$} (ideal circuit $C_1$), is affected by a single bit-flip error (erroneous implementation $C_2$) in the second layer. It is easy to check that only~4 out of all $2^8=256$ input strings can detect this discrepancy.

\begin{figure}[t]
    \centering
    \resizebox{0.9\linewidth}{!}{
    \begin{tikzpicture}[baseline,scale=0.5]
    \foreach \y in {-1,0,...,2}
    {\draw (-0.45,\y-0.45) rectangle (0.45,\y+0.45);
    \node at (0,\y) {$\wedge$};
    \draw (-2,\y-0.25) -- (-0.45,\y-0.25);
    \draw (-2,\y+0.25) -- (-0.45,\y+0.25);
    \draw (0.45,\y) -- (2,\y);
    };
    \fill[white] (0.75,0.5) rectangle (1.75,1.5);
    \fill[red,opacity=0.2] (0.75,0.5) rectangle (1.75,1.5);
    \draw (0.75,0.5) rectangle (1.75,1.5);
    \node at (1.25,1) {$\neg$};
    \draw (2,0.75) rectangle (3.5,2.25);
    \node at (2.75,1.5) {$\wedge$};
    \draw (2,-1.25) rectangle (3.5,0.25);
    \node at (2.75,-0.5) {$\wedge$};
    \draw (3.5,1.5) -- (5.5,1.5);
    \draw (3.5,-0.5) -- (5.5,-0.5);
    \draw (5.5,-0.75) rectangle (8,1.75);
    \node at (6.75,0.5) {$\wedge$};
    \draw (8,0.5) -- (10,0.5);
    \foreach \y in {1,2,...,8}
    {\node at (-2.5,2.75-0.5*\y) {$x_\y$};
    };
    \node at (10.5,0.5) {$y'$};
    \end{tikzpicture}}
    \caption{\emph{Error detection in classical circuits is hard:} Suppose that a cascade of logical \textsc{AND} gates, realizing the Boolean function $y=x_8\cdot\hdots\cdot x_1$, 
    is affected by a single bit-flip error (red) in the second layer. 
    Only $4$ out of the $2^8=256$ possible input strings can detect this error.}
    \label{fig:irreversible}
\end{figure}
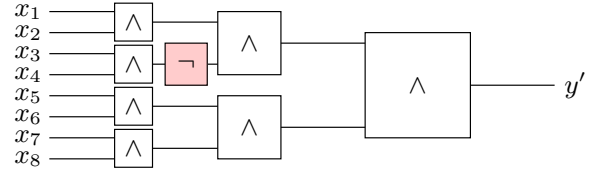

Masking is a serious issue for error detection using simulation techniques.
No malicious intent is required to fool randomly generated inputs. The circuit may do it all by itself. 
Needless to say, this issue has been well-known for decades.
Error detection based on random inputs (alone) often pales in comparison to other more sophisticated techniques. 
Today's state of the art is governed by constrained-based stimuli generation techniques~\cite{yuanConstraintbasedVerification2006, biereSATATPGBoolean2002, willeSMTbasedStimuliGeneration2009, kitchenStimulusGenerationConstrained2007, gentFastMultilevelTest2016}, fuzzing~\cite{laeuferRFUZZCoveragedirectedFuzz2018}, etc. 
But on the positive side, error detection using randomly-chosen inputs is based on minimal assumptions, namely the possibility to simulate two circuits as black boxes.
Moreover, it is intuitive and individual simulation runs are easy and fast to execute. 

\section{Summary of results: \\Error detection in reversible circuits}

We have seen that, in general, simulation with (uniformly) random inputs is not a viable strategy for detecting errors in classical circuits. Already a single ``small'' error can be exceedingly difficult to detect (masking). 
Perhaps surprisingly, this dark picture lightens up considerably if we consider \emph{reversible} implementations of logical functionalities. 
As the name suggests, reversible circuits are circuits whose action can be undone by running the circuit backwards. More formally, $n$-bit reversible circuits implement permutations on the set of all $2^n$ bit strings.
This, in particular, implies that the number of input and output bits must be the same ($n=m$). Despite these restrictions, reversible circuits are universal, i.e., \emph{any} logical function on $n$ bits can be implemented by a reversible circuit~\cite{toffoliReversibleComputing1980} and efficient mapping techniques are readily available~\cite{zulehnermakeitreversible2017, maslovReversibleCascadesMinimal2004, zilicReversibleCircuitTechnology2007}
(this implementation may require strictly more than $n$ bits, though). Negation ($\textsc{NOT}$), exclusive or (\textsc{CNOT}) and the Toffoli gate (\textsc{CCNOT}) are examples of simple reversible functionalities. 
Viewed as a logic gate, $\textsc{CCNOT}$ is also universal. Every reversible circuit can be constructed from Toffoli gates alone~\cite{toffoliReversibleComputing1980}.
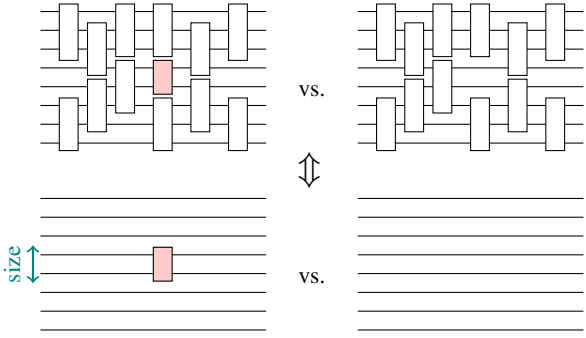
\begin{figure}
    \centering
\begin{tabular}{rcl}
\begin{tikzpicture}[baseline,scale=0.5]
    \foreach \y in {-1.25,-0.75,...,2.25}
    {\draw (-3,\y) -- (3,\y);
    };
    \draw[fill=white] (-2.5,0.95) rectangle (-2,2.45);
    \draw[fill=white] (-2.5,-1.45) rectangle (-2,-0.05);
    \draw[fill=white] (-1.75, -0.95) rectangle (-1.25,0.45);
    \draw[fill=white] (-1.75,0.55) rectangle (-1.25,1.95);
    \draw[fill=white] (-1,1.05) rectangle (-0.5,2.45);
    \draw[fill=white] (-1,-0.45) rectangle (-0.5,0.95);
    \draw[fill=white] (0,0.05) rectangle (0.5,0.95);
    \draw[fill=red,opacity=0.2] (0,0.05) rectangle (0.5,0.95);
    \draw[fill=white] (0,-1.45) rectangle (0.5,-0.05);
    \draw[fill=white] (0,1.05) rectangle (0.5,2.45);
    \draw[fill=white] (1,-0.95) rectangle (1.5,0.45);
    \draw[fill=white] (1,0.55) rectangle (1.5,1.95);
    \draw[fill=white] (2,0.95) rectangle (2.5,2.45);
    \draw[fill=white] (2,-1.45) rectangle (2.5,-0.05);
\end{tikzpicture}
& vs.\ &
\begin{tikzpicture}[baseline,scale=0.5]
    \foreach \y in {-1.25,-0.75,...,2.25}
    {\draw (-3,\y) -- (3,\y);
    };
    \draw[fill=white] (-2.5,0.95) rectangle (-2,2.45);
    \draw[fill=white] (-2.5,-1.45) rectangle (-2,-0.05);
    \draw[fill=white] (-1.75, -0.95) rectangle (-1.25,0.45);
    \draw[fill=white] (-1.75,0.55) rectangle (-1.25,1.95);
    \draw[fill=white] (-1,1.05) rectangle (-0.5,2.45);
    \draw[fill=white] (-1,-0.45) rectangle (-0.5,0.95);
    \draw[fill=white] (0,-1.45) rectangle (0.5,-0.05);
    \draw[fill=white] (0,1.05) rectangle (0.5,2.45);
    \draw[fill=white] (1,0.55) rectangle (1.5,1.95);
    \draw[fill=white] (1,-0.95) rectangle (1.5,0.45);
    \draw[fill=white] (2,0.95) rectangle (2.5,2.45);
    \draw[fill=white] (2,-1.45) rectangle (2.5,-0.05);
\end{tikzpicture} \\
& 
\rotatebox{90}{\LARGE $\Leftrightarrow$}
& \\
\begin{tikzpicture}[baseline,scale=0.5]
    \foreach \y in {-1.25,-0.75,...,2.25}
    {\draw (-3,\y) -- (3,\y);
    };
    \draw[fill=white] (0,0.05) rectangle (0.5,0.95);
    \draw[fill=red,opacity=0.2] (0,0.05) rectangle (0.5,0.95);
    \draw[teal,<->,thick] (-3.2,0) -- (-3.2,1);
    \node[rotate=90] at (-3.75,0.5) {\textcolor{teal}{size}};
\end{tikzpicture}
& vs. &
\begin{tikzpicture}[baseline,scale=0.5]
    \foreach \y in {-1.25,-0.75,...,2.25}
    {\draw (-3,\y) -- (3,\y);
    };
\end{tikzpicture}
\end{tabular}
    \caption{\emph{Illustration of main rigorous contributions:}
    Simulations with uniformly random inputs completely expose any single error in a given reversible circuit. The two scenarios are \emph{exactly} equivalent (``no masking''). In the lower scenario, the probability of correct distinction is governed by the size $k$ of the error, not the total number of lines.
    }
    \label{fig:theory}
\end{figure}

To summarize, reversible circuits bear strong similarities with classical (irreversible) circuits, but there are some notable additional characteristics. Chief among them is reversibility itself which implies that information cannot easily escape. 
Here, we show that this has profound implications for error detection with random inputs. More precisely, 
\begin{enumerate}[(i)]
    \item reversible circuits can never mask single errors (rigorous result, see Proposition~\ref{prop:no-masking})
    \item the probability of detecting a single error only depends on its size, not the total number of bits (unsurprising rigorous result, see Lemma~\ref{lem:error-size})
    \item multiple errors are typically even easier to detect (empirical studies, see Fig.~\ref{fig:numerics} and discussions in Section~\ref{sec:numerics})
\end{enumerate}
The first two insights are mathematical statements that address single errors only. They readily follow from reversibility and fundamental properties of uniformly random input strings. We refer to Section~\ref{sec:theory} for details and Fig.~\ref{fig:theory} for illustrative caricatures. 
When combined, they imply the following confidence bound for detecting single errors with random inputs.

\begin{theorem} \label{thm:main}
Suppose that a general reversible circuit is affected by a single error of size $k$ and fix $\delta \in (0,1)$ (confidence). Then, at most $\lceil \log (1/\delta) 2^{k-1} \rceil$ randomly selected inputs suffice to witness this error with probability (at least) $1-\delta$.
\end{theorem}

For $k=1$---a single bit-flip error (\textsc{NOT}) \emph{anywhere} within the circuit---this statement actually becomes deterministic: already a single (random) input is guaranteed to detect this error with certainty. We emphasize that this statement is true irrespective of the number of lines and the circuit's size. It is simply impossible to hide a single bit-flip inside a reversible circuit. Such a behavior is strikingly different from irreversible circuit architectures. There it can routinely happen that order~$2^n$ random inputs are necessary to detect even a single bit-flip error, see e.g.\ Fig.~\ref{fig:irreversible}.

The multiple-error case is much more intricate,
because error locations and circuit structure start to matter. 
This leads to drastically different behaviors of best case (independent errors) and worst case (severe masking) behavior. To better understand the typical behavior of multiple errors, we resort to numerical simulations.
These indicate a (close-to) best-case behavior: the probability of failing to detect a total of $l$ errors is exponentially suppressed in $l$, see Fig.~\ref{fig:numerics}. Additional simulation results and details are provided in Section~\ref{sec:numerics}.

\begin{figure}
    \centering
    \includegraphics[width=0.99\linewidth]{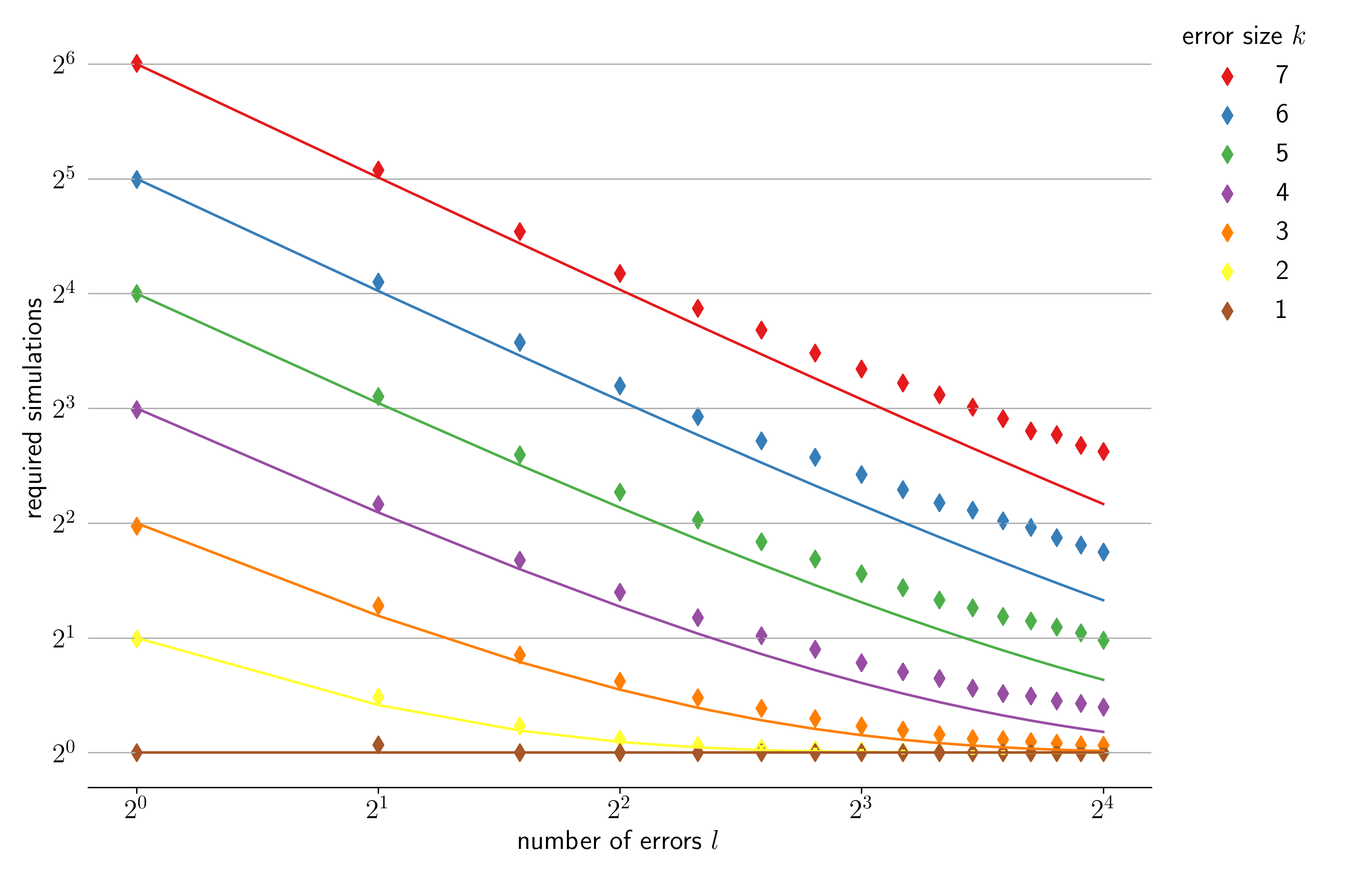}
\caption{\emph{Typical accumulation effects for multiple errors (log-log plot):} 
number~$l$ of randomly injected errors ($x$-axis) vs.\ average number of random inputs required to detect erroneous behavior ($y$-axis) in a generic \mbox{$n=20$-bit} reversible circuit with $4000$ 
gates. Different colors denote worst-case errors of increasing size $k$. 
Solid lines track the theoretical best-case behavior (independent errors, see Eq.~\eqref{eq:best-case-expectation} below).
For small $l$, the plot highlights an excellent agreement between typical (diamonds) and best-case (solid lines) behavior.}
   \label{fig:numerics}
   \vspace{-3mm}
\end{figure}




Note that a similar line of thought has recently been presented for the domain of quantum computing (which bears many similarities to reversible circuits). More precisely, a verification scheme heavily based on simulation has been proposed in~\cite{burgholzerPowerSimulationEquivalence2020} and refined in~\cite{burgholzerRandomStimuliGeneration2021}. A similar theoretical result has been presented in~\cite{lindenLightweightDetectionSmall2020}.

\vspace*{-4mm}
\section{Rigorous theory for single errors} \label{sec:theory}

\subsection{Reversible circuits and error model}

We will work in the reversible circuit model for $n$ input bits (and $n$ output bits). A high-level of mathematical abstraction already suffices to deduce powerful consequences. An $n$-bit reversible circuit implements a permutation \mbox{$R:\left\{0,1\right\}^n \to \left\{0,1\right\}^n$} of all $2^n$ bit strings. Reversing the circuit, that is running it backwards, produces the unique permutation $R^T: \left\{0,1\right\}^n\to \left\{0,1\right\}^n$ that undoes the original circuit: $R^T\circ R = R \circ R^T=\mathrm{id}$, where $\mathrm{id}(\vec{x})=\vec{x}$ for all $\vec{x} \in \left\{0,1\right\}^n$ is the identity permutation (``do nothing''). 
This defining feature suffices to deduce three elementary properties that will form the basis of our proof strategy.

\begin{lemma}[Characteristics of reversible circuits] \label{lem:reversible}
Consider  reversible circuits $R_1,R_2,R_3: \left\{0,1\right\}^n \to \left\{0,1\right\}^n$ and an $n$-bit string \mbox{$\vec{x}\in \left\{0,1\right\}^n$}. Then,
\begin{enumerate}[(i)]
\item output equivalence is unaffected by composition: \\
$R_1 (\vec{x}) = R_2 (\vec{x}) 
\Leftrightarrow (R_3 \circ R_1) (\vec{x}) = (R_3 \circ R_2) (\vec{x})$
\item invariance of the uniform distribution:  \\$\vec{x} \sim \textsc{Unif}(\left\{0,1\right\}^n)$ implies $R_1(\vec{x}) \sim \textsc{Unif}(\left\{0,1\right\}^n)$
\item non-trivial action: 
suppose $R_1 \neq \mathrm{id}$. Then, there are at least two bit strings 
such that $R_1 (\vec{x}) \neq \vec{x}$.
\end{enumerate}
\end{lemma}

\begin{proof}
All proofs utilize
the fact that reversible circuits act like permutations on the set of all $2^n$ bit strings.
\begin{enumerate}[(i)]
    \item Permutations are \emph{invertible} transformations. As such, they preserve equivalence: $y=y'$ if and only if \mbox{$R(y) = R(y')$} for any reversible circuit $R$. The claim follows from setting \mbox{$y=R_1(\vec{x})$}, \mbox{$y' = R_2 (\vec{x})$} and $R=R_3$.
    \item The uniform distribution over $n$-bit strings assigns the same weight to each of the $2^n$ bit strings. Permuting the bit strings cannot affect the weights and, by extension, the uniform distribution itself.
    \item The number of invariant bit strings ($\mbox{$\vec{x}\in\left\{0,1\right\}^n\colon$} R_1 (\vec{x})=\vec{x}$) is equal to the number of fix points of the underlying permutation. A non-trivial permutation of $2^n$ elements can have at most $2^n-2$ fix points (transposition).
\end{enumerate}
\end{proof}

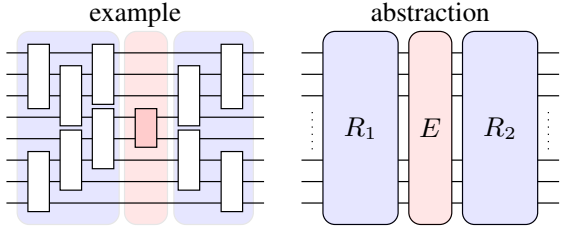
\begin{figure}
\centering
\resizebox{0.9\linewidth}{!}{
\begin{tabular}{cc}
   example & abstraction \\
    \begin{tikzpicture}[baseline,scale=0.5]
    \draw[rounded corners,fill=blue,opacity=0.1] (-2.75,-1.75)
    rectangle (-0.375,2.75);
    \draw[rounded corners, fill=red,opacity=0.1]
    (-0.25,-1.75) rectangle (0.75,2.75);
    \draw[rounded corners,fill=blue,opacity=0.1] (0.9,-1.75)
    rectangle (2.75,2.75);
    \foreach \y in {-1.25,-0.75,...,2.25}
    {\draw (-3,\y) -- (3,\y);
    };
    \draw[fill=white] (-2.5,0.95) rectangle (-2,2.45);
    \draw[fill=white] (-2.5,-1.45) rectangle (-2,-0.05);
    \draw[fill=white] (-1.75, -0.95) rectangle (-1.25,0.45);
    \draw[fill=white] (-1.75,0.55) rectangle (-1.25,1.95);
    \draw[fill=white] (-1,1.05) rectangle (-0.5,2.45);
    \draw[fill=white] (-1,-0.45) rectangle (-0.5,0.95);
    \draw[fill=white] (0,0.05) rectangle (0.5,0.95);
    \draw[fill=red,opacity=0.2] (0,0.05) rectangle (0.5,0.95);
    \draw[fill=white] (1,-0.95) rectangle (1.5,0.45);
    \draw[fill=white] (1,0.55) rectangle (1.5,1.95);
    \draw[fill=white] (2,0.95) rectangle (2.5,2.45);
    \draw[fill=white] (2,-1.45) rectangle (2.5,-0.05);
    \end{tikzpicture}
    &
   \begin{tikzpicture}[baseline,scale=0.5]
    \foreach \y in {-1.25,-0.75,-0.25,1.25,1.75,2.25}
    {\draw (-3,\y) -- (3,\y);
    };
    \draw[dotted] (-2.75,0) -- (-2.75,1);
    \draw[dotted] (2.75,0) -- (2.75,1);
    \draw[rounded corners,fill=white] (-2.5,-1.75) rectangle (-0.75,2.75);
    \fill[rounded corners,blue,opacity=0.1] (-2.5,-1.75) rectangle (-0.75,2.75);
    \node at (-1.625,0.5) {$R_1$};
    \draw[rounded corners,fill=white] (-0.5,-1.75) rectangle (0.5,2.75);
    \fill[rounded corners,red,opacity=0.1] (-0.5,-1.75) rectangle (0.5,2.75);
    \node at (0,0.5) {$E$};
    \draw[rounded corners,fill=white] (0.75,-1.75) rectangle (2.5,2.75);
    \fill[rounded corners,blue,opacity=0.1] (0.75,-1.75) rectangle (2.5,2.75);
    \node at (1.625,0.5) {$R_2$};
    \end{tikzpicture} 
    \end{tabular}}
    \caption{\emph{(Single) error model and compatible circuit decomposition:} 
    An ideal reversible circuit (blue) is corrupted by a single reversible error (red). The error location begets a decomposition of ideal and corrupted circuit into matching constituents: $R=R_2 \circ R_1$ (ideal) and $\tilde{R}=R_2 \circ E \circ R_1$ (corrupted).}
    \label{fig:error-decomposition}
\end{figure}

Different reversible circuits of compatible bit-size~$n$ can be combined to yield another (larger) circuit: \mbox{$(R_2 \circ R_1) (\vec{x})=R_2\left( R_1 (\vec{x})\right)$} for input $\vec{x} \in \left\{0,1\right\}^n$ (``composition''). The reverse direction is also possible (``decomposition'') and, arguably, more interesting. Circuit diagrams provide a well-established tool that does precisely that. They decompose a possibly complicated circuit into a structured sequence of simpler building blocks.
We use circuit decomposition on a rather high level to reason about \emph{single} errors in reversible circuits. 
Suppose that an $n$-bit reversible circuit $R$ is affected by a reversible error $E$ that produces a functionally different circuit $\tilde{R}$.
Then, the location of this error within the circuit suggests a compatible decomposition into three parts:
\begin{enumerate}[(i)]
    \item $R_1: \left\{0,1\right\}^n \to \left\{0,1\right\}^n$ describes the original functionality up to the location where the error occurs (``past''),
    \item $E:\left\{0,1\right\}^n \to \left\{0,1\right\}^n$ captures the error as an additional circuit layer on all $n$ bits (``present''),
    \item $R_2: \left\{0,1\right\}^n \to \left\{0,1\right\}^n$ describes the original functionality from the error location onwards (``future'').
\end{enumerate}
In summary,
\begin{equation}
\tilde{R}= R_2 \circ E \circ R_1,
\quad \text{while} \quad R = R_2 \circ R_1,
\label{eq:single-error-model}
\end{equation}
and we refer to Fig.~\ref{fig:error-decomposition} for a visual illustration.

\subsection{No masking for random inputs}

We now have all building blocks in place to present and derive the main conceptual result of this work. 
It addresses the probability of detecting single errors in arbitrary reversible circuits \eqref{eq:single-error-model} based on a single random input $\vec{x} \sim \mathrm{Unif}(\left\{0,1\right\}^n)$.

\begin{proposition}[no masking] \label{prop:no-masking}
Fix $R=R_2 \circ R_1$ (ideal circuit) and $\tilde{R}=R_2 \circ E \circ R_1$ (single, reversible error).
Then, the probability of detecting this discrepancy with a random input $\vec{x}\sim \mathrm{Unif}(\left\{0,1\right\}^n)$ 
only depends on the error $E$, not the actual circuit. More precisely,
\begin{equation*}
\mathrm{Pr}
\left[
\tilde{R}(\vec{x}) \neq R (\vec{x}) \right]
=\mathrm{Pr}
\left[E(\vec{x})\neq \vec{x}\right],
\end{equation*}
where the probability is taken with respect to the uniform distribution over all $2^n$ possible input strings.

\end{proposition}

\begin{proof}
This statement is an immediate consequence of two elementary characteristics of reversible circuit architectures. 
Apply Lemma~\ref{lem:reversible} (i)
to remove the effect of $R_2$,
\begin{align*}
\mathrm{Pr}\left[ \tilde{R}(\vec{x}) = R(\vec{x}) \right]
=& \left[ R_2 \circ E \circ R_1 (\vec{x}) = R_2 \circ R_1 (\vec{x}) \right] \\
=& \,\mathrm{Pr} \left[ E \left( R_1 (\vec{x}) \right) = R_1 (\vec{x}) \right],
\end{align*}
and note that, according to Lemma~\ref{lem:reversible} (ii), \mbox{$\vec{x} \sim \mathrm{Unif}(\left\{0,1\right\}^n)$} implies \mbox{$R_1 (\vec{x}) \sim \mathrm{Unif}(\left\{0,1\right\}^n)$}.
\end{proof}
\vspace{200cm}

Although simple to prove, Proposition~\ref{prop:no-masking} pinpoints remarkable differences between reversible and irreversible circuits.
As illustrated in Fig.~\ref{fig:theory}, the former cannot hide errors from randomly sampled inputs (``no masking'').

We emphasize that a uniformly random selection of input strings is crucial to arrive at such a powerful conclusion. 
Reversibility alone is enough to ignore the final portion of the circuit $R_2$ (after the error has occurred). Reversible circuits always map (non-)equal bit strings to (non-)equal bit strings. In contrast, the first portion of the circuit $R_1$ (before the error has occurred) can affect concrete inputs
$\vec{x} \in \left\{0,1\right\}^n$. 
But if $\vec{x}$ is sampled randomly, then $R_1 (\vec{x})$ will be
a different, but still random, bit string. The uniform distribution is special in the sense that it is invariant under reversible transformations.
The circuit $R_1$ may affect every concrete input, but it does not affect the underlying distribution.

\subsection{Only error size matters}\label{sub:size}

We have seen that uniformly random inputs can uncover single errors in a general reversible circuit. According to Proposition~\ref{prop:no-masking}, the probability of witnessing a discrepancy only depends on the error, not the underlying circuit structure.

We say that an error $E:\left\{0,1\right\}^n \to \left\{0,1\right\}^n$ has \emph{size} $k$ if it only affects $k$ bits in a nontrivial fashion. The remaining $n-k$ bits are not touched at all. We refer to Fig.~\ref{fig:theory} for a visual illustration of this summary parameter. 
Intuitively, we would expect that ``large'' errors are easier to detect than ``small'' ones and that the number of lines $n$ plays an active role. 
However, the following simple statement
 shows that the probability of detecting an error in the worst case is exponentially suppressed with respect to the error size $k$, but is independent of the actual number of bits $n$.

\begin{lemma}[only error size matters]
\label{lem:error-size}
Suppose that \mbox{$E: \left\{0,1\right\}^n \to \left\{0,1\right\}^n$} is a non-trivial error that only affects $k$ bits in a non-trivial fashion 
and 
$\vec{x} \sim \mathrm{Unif}(\left\{0,1\right\}^n)$ is sampled from the uniform distribution. Then,
\begin{align*}
\mathrm{Pr}\left[ E(\vec{x}) \neq \vec{x}\right] \geq 2^{-(k-1)}.
\end{align*}
\end{lemma}

\begin{proof}
Suppose, without loss of generality, that the error $E$ only affects the least-significant $k$ bits, i.e., \mbox{$E(\vec{x})=E(x_n,\ldots,x_1)=(x_n,\ldots,x_{k+1},y_{k},\ldots,y_1)$}, where $(y_k,\ldots,y_1)=\tilde{E}(x_k,\ldots,x_1)$. Since $E$ is reversible, its restriction $\tilde{E}: \left\{0,1\right\}^k \to \left\{0,1\right\}^k$ to the $k$ relevant bits must also be reversible. Moreover, $\tilde{E} \neq \mathrm{id}$, because $E$ is non-trivial. 
Lemma~\ref{lem:reversible} (iii) then implies that there must be at least 2 bit strings of size $k$ that are affected by $\tilde{E}$. 
Finally, we use the fact that $\vec{x}=(x_n,\ldots,x_1) \sim \mathrm{Unif}(\left\{0,1\right\}^n)$ implies that the least-significant $k$ bits are also distributed uniformly: $(x_k,\ldots,x_1) \sim \mathrm{Unif}(\left\{0,1\right\}^k)$.
Therefore,
\begin{align*}
\mathrm{Pr}\left[ E(\vec{x})\neq \vec{x}\right]
= \mathrm{Pr}\left[ \tilde{E}(x_k,\ldots,x_1) \neq (x_k,\ldots,x_1) \right]
\geq \frac{2}{2^k}.
\end{align*}
\end{proof}

This probability bound is actually sharp. Worst-case errors of size $k$ permute exactly 2 out of the $2^k$ possible $k$-bit inputs on which they act. Concrete examples of such a behavior are $\textsc{NOT}$ ($k=1$), $\textsc{CNOT}$ ($k=2$), $\textsc{CCNOT}$ ($k=3$) and, more generally, a $(k-1)$-fold controlled \textsc{NOT} gate on $k$ bits (general $k$). 
The numerical simulations shown in Fig.~\ref{fig:numerics} are based on injecting such worst-case errors at random circuit locations.

\subsection{General confidence bound for detecting single errors}

We now have all necessary ingredients to establish a rigorous performance guarantee for reversible error detection with (uniformly) random inputs. 
The following statement bounds the number of uniformly random inputs that may be required to detect a single error of size $k$.

\begin{theorem} \label{thm:main-restatement}
Fix $R=R_2 \circ R_1$ (ideal circuit), $\tilde{R}=R_2 \circ E \circ R_1$ (single error) and $E$ has size $k$.
Suppose that $\vec{x}_1,\ldots,\vec{x}_N$ are $N$ (independent) uniformly random inputs. Then,
\begin{equation*}
\mathrm{Pr}\Big[ \bigwedge_{1 \leq i \leq N} \left\{\tilde{R}(\vec{x}_i) = R(\vec{x}_i) \right\} \Big]
\leq \exp \left( - N/2^{k-1}\right)
\end{equation*}
In words, the probability of failing to detect a single error is exponentially suppressed in the number $N$ of random test inputs. 
\end{theorem}

Theorem~\ref{thm:main} above is a streamlined consequence of this observation: setting $N=\lceil \log(1/\delta)2^{k-1}\rceil $ provides a concrete number of repetitions that ensures that we detect the discrepancy with probability (at least) $1-\delta$.

\begin{proof}[Proof of Theorem~\ref{thm:main-restatement}]
For $N=1$ (one random input), the claim readily follows from combining Proposition~\ref{prop:no-masking} and Lemma~\ref{lem:error-size} (more precisely, their contrapositions):
\begin{align*}
\mathrm{Pr}\left[ \tilde{R}(\vec{x}_1) = R(\vec{x}_1)\right]
= \mathrm{Pr}\left[ E(\vec{x}_1) = \vec{x}_1 \right] \leq 1- 2^{-(k-1)}.
\end{align*}
This bound readily extends to the general $N$-case by using the assumption that the individual input strings $\vec{x}_1,\ldots,\vec{x}_N$ are all sampled independently. Joint probabilities of independent events factorize and we conclude
\begin{align}
\mathrm{Pr}\Big[ \bigwedge_{1 \leq i \leq N} \left\{\tilde{R}(\vec{x}_i) = R(\vec{x}_i) \right\} \Big] \nonumber
=& \prod_{i=1}^N 
\mathrm{Pr}\left[ \tilde{R}(\vec{x}_i) = R(\vec{x}_i)\right] \\
\leq & \left( 1- 2^{-(k-1)}\right)^N.
\label{eq:aux1}
\end{align}
Apply $1+x < \exp (x)$ for all $x \in \mathbb{R}$ (convexity of the exponential function) to complete the argument.
\end{proof}

The bound provided in Theorem~\ref{thm:main-restatement} is simple, but not sharp (the inequality $1+x \leq \exp (x)$ is never tight).  As such, it always under-estimates the actual confidence level. This discrepancy is most pronounced for small error sizes $k$. 
The extreme case is a single \textsc{NOT} error ($k=1$). For $k=1$, the bound in Eq.~\eqref{eq:aux1} becomes (exactly) zero. By contraposition, \emph{every possible input bit string is guaranteed to detect a single bit-flip error that is hidden anywhere within the circuit.}

\section{Empirical analysis for multiple errors} \label{sec:numerics}

In the previous section, we have established strong theoretical support for detecting single errors. At its heart has been the decomposition $\tilde{R}=R_2 \circ E \circ R_1$ illustrated in Fig.~\ref{fig:error-decomposition}.
Reversibility and uniformly random inputs have subsequently allowed us to discuss away the circuit portions $R_2$ and $R_1$ completely. In turn, we were able to focus exclusively on the error itself.

\begin{figure}
    \centering
    \resizebox{0.99\linewidth}{!}{
 \begin{tabular}{rcl}
\begin{tikzpicture}[baseline,scale=0.5]
    \foreach \y in {-1.25,-0.75,-0.25,1.25,1.75,2.25}
    {\draw (-4.5,\y) -- (4.5,\y);
    };
    \draw[dotted] (-4.25,0) -- (-4.25,1);
    \draw[dotted] (4.25,0) -- (4.25,1);
    \draw[rounded corners, fill=white] 
    (-4,-1.75) rectangle (-2.5,2.75);
    \fill[rounded corners, blue, opacity=0.1] 
    (-4,-1.75) rectangle (-2.5,2.75);
    \node at (-3.25,0.5) {$R_1$};
    \draw[rounded corners, fill=white] 
    (-2.25,-1.75) rectangle (-1.25,2.75);
    \fill[rounded corners, red, opacity=0.1] 
    (-2.25,-1.75) rectangle (-1.25,2.75);
    \node at (-1.75,0.5) {$E_1$};
    \draw[rounded corners, fill=white] 
    (-1,-1.75) rectangle (1,2.75);
    \fill[rounded corners, blue, opacity=0.1] 
    (-1,-1.75) rectangle (1,2.75);
    \node at (0,0.5) {$R_2$};
    \draw[rounded corners, fill=white] 
    (1.25,-1.75) rectangle (2.25,2.75);
    \fill[rounded corners, red, opacity=0.1] 
    (1.25,-1.75) rectangle (2.25,2.75);
    \node at (1.75,0.5) {$E_2$};
    \draw[rounded corners, fill=white] 
    (2.5,-1.75) rectangle (4,2.75);
    \fill[rounded corners, blue, opacity=0.1] 
    (2.5,-1.75) rectangle (4,2.75);
    \node at (3.25,0.5) {$R_3$};
    \end{tikzpicture} 
    & vs & 
    \begin{tikzpicture}[baseline,scale=0.5]
    \foreach \y in {-1.25,-0.75,-0.25,1.25,1.75,2.25}
    {\draw (-3.5,\y) -- (3.5,\y);
    };
    \draw[dotted] (-3.25,0) -- (-3.25,1);
    \draw[dotted] (3.25,0) -- (3.25,1);
    \draw[rounded corners, fill=white] 
    (-3,-1.75) rectangle (-1.5,2.75);
    \fill[rounded corners, blue,opacity=0.1] 
    (-3,-1.75) rectangle (-1.5,2.75);
    \node at (-2.25,0.5) {$R_1$};
    \draw[rounded corners, fill=white] 
    (-1,-1.75) rectangle (1,2.75);
    \fill[rounded corners, blue,opacity=0.1] 
    (-1,-1.75) rectangle (1,2.75);
    \node at (0,0.5) {$R_2$};
    \draw[rounded corners, fill=white] 
    (1.5,-1.75) rectangle (3,2.75);
    \fill[rounded corners, blue,opacity=0.1] 
    (1.5,-1.75) rectangle (3,2.75);
    \node at (2.25,0.5) {$R_3$};
    \end{tikzpicture} \\
    & \rotatebox{90}{\LARGE $\Leftrightarrow$} & \\
    \begin{tikzpicture}[baseline,scale=0.5]
    \foreach \y in {-1.25,-0.75,-0.25,1.25,1.75,2.25}
    {\draw (-2.75,\y) -- (2.75,\y);
    };
    \draw[dotted] (-2.5,0) -- (-2.5,1);
    \draw[dotted] (2.5,0) -- (2.5,1);
    \draw[rounded corners, fill=white] 
    (-2.25,-1.75) rectangle (-1.25,2.75);
    \fill[rounded corners, red,opacity=0.1] 
    (-2.25,-1.75) rectangle (-1.25,2.75);
    \node at (-1.75,0.5) {$E_1$};
    \draw[rounded corners, fill=white] 
    (-1,-1.75) rectangle (1,2.75);
    \fill[rounded corners, blue, opacity=0.1]
    (-1,-1.75) rectangle (1,2.75);
    \node at (0,0.5) {$R_2$};
    \draw[rounded corners, fill=white] 
    (1.25,-1.75) rectangle (2.25,2.75);
    \fill[rounded corners, red, opacity=0.1] 
    (1.25,-1.75) rectangle (2.25,2.75);
    \node at (1.75,0.5) {$E_2$};
    \end{tikzpicture} 
    & vs & 
    \begin{tikzpicture}[baseline,scale=0.5]
    \foreach \y in {-1.25,-0.75,-0.25,1.25,1.75,2.25}
    {\draw (-1.5,\y) -- (1.5,\y);
    };
    \draw[dotted] (-1.25,0) -- (-1.25,1);
    \draw[dotted] (1.25,0) -- (1.25,1);
    \draw[rounded corners, fill=white] 
    (-1,-1.75) rectangle (1,2.75);
    \fill[rounded corners, blue, opacity=0.1] 
    (-1,-1.75) rectangle (1,2.75);
    \node at (0,0.5) {$R_2$};
    \end{tikzpicture} 
    \end{tabular}}
    \caption{\emph{Partial simplification for multiple errors:} Simulation with uniformly random inputs exposes multiple errors only partially. Everything before the first error ($R_1$) and after the last error ($R_3$) can be safely ignored, but the part in between ($R_2$) does matter. Different circuit structures can lead to strikingly different error detection probabilities.}
    \label{fig:multiple-errors}
\end{figure}
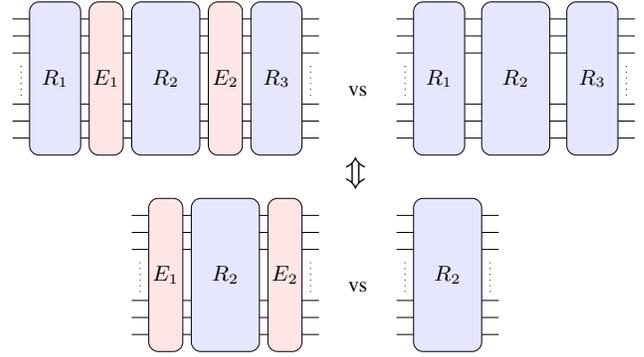

For more than one error, 
this 
is in general not an option anymore. 
While we can safely ignore circuit contributions before the first and after the last error, the circuit in between cannot be ignored, see Fig.~\ref{fig:multiple-errors}.  
The relation between errors and intermediate circuit parts governs how likely it is to witness the overall error.

In this section, we analyze error accumulation effects in generic reversible circuits. 
To obtain guiding intuition, we will first isolate and discuss the two extreme cases. Independent errors (best case, see Section~\ref{sub:best-case}) and maximal masking (worst case, see Section~\ref{sub:worst-case})
turn out to behave in a radically different fashion.
Subsequent numerical studies demonstrate that typical error accumulation effects closely follow the best-case trajectory: Multiple errors are typically \emph{much} easier to detect than a single error.

\subsection{Best-case behavior: Commuting and independent errors} \label{sub:best-case}

\begin{figure}
    \centering
\begin{align*}
\begin{tikzpicture}[baseline,scale=0.5]
    \foreach \y in {-1.25,-0.75,-0.25,1.25,1.75,2.25}
    {\draw (-2.75,\y) -- (2.75,\y);
    };
    \draw[dotted] (-2.5,0) -- (-2.5,1);
    \draw[dotted] (2.5,0) -- (2.5,1);
    \draw[rounded corners, fill=white] 
    (-2.25,-1.75) rectangle (-1.25,2.75);
    \fill[rounded corners, red,opacity=0.1] 
    (-2.25,-1.75) rectangle (-1.25,2.75);
    \node at (-1.75,0.5) {$E_1$};
    \draw[rounded corners, fill=white] 
    (-1,-1.75) rectangle (1,2.75);
    \fill[rounded corners, blue, opacity=0.1]
    (-1,-1.75) rectangle (1,2.75);
    \node at (0,0.5) {$R_2$};
    \draw[rounded corners, fill=white] 
    (1.25,-1.75) rectangle (2.25,2.75);
    \fill[rounded corners, red, opacity=0.1] 
    (1.25,-1.75) rectangle (2.25,2.75);
    \node at (1.75,0.5) {$E_2$};
    \end{tikzpicture} 
     = & 
    \begin{tikzpicture}[baseline,scale=0.5]
    \foreach \y in {-1.25,-0.75,-0.25,1.25,1.75,2.25}
    {\draw (-2.75,\y) -- (2.75,\y);
    };
    \draw[dotted] (-2.5,0) -- (-2.5,1);
    \draw[dotted] (2.5,0) -- (2.5,1);
    (-2.25,-1.75) rectangle (-1.25,2.75);
    \fill[rounded corners, red,opacity=0.1] 
    (-2.25,-1.75) rectangle (-1.25,2.75);
    \node at (-1.75,0.5) {$E_1$};
    \draw[fill=white] (-2,1.05) rectangle (-1.5,2.45);
    \fill[red,opacity=0.2] (-2,1.05) rectangle (-1.5,2.45);
    \fill[rounded corners, red, opacity=0.1] 
    (-1,-1.75) rectangle (0,2.75);
    \node at (-0.5,0.5) {$E_2$};
    \draw[fill=white] (-0.75,-1.45) rectangle (-0.25,-0.05);
    \fill[red,opacity=0.2] (-0.75,-1.45) rectangle (-0.25,-0.05);
    \draw[rounded corners, fill=white] 
    (0.25,-1.75) rectangle (2.25,2.75);
    \fill[rounded corners, blue, opacity=0.1]
    (0.25,-1.75) rectangle (2.25,2.75);
    \node at (1.25,0.5) {$R_2$};
    \end{tikzpicture} 
    =
    \begin{tikzpicture}[baseline,scale=0.5]
    \foreach \y in {-1.25,-0.75,-0.25,1.25,1.75,2.25}
    {\draw (-1.5,\y) -- (2.75,\y);
    };
    \draw[dotted] (-1.25,0) -- (-1.25,1);
    \draw[dotted] (1.25,0) -- (1.25,1);
    \fill[rounded corners, red, opacity=0.1] 
    (-1,-1.75) rectangle (0,2.75);
    \node at (-0.5,0.5) {$\tilde{E}$};
    \draw[fill=white] (-0.75,-1.45) rectangle (-0.25,-0.05);
    \fill[red,opacity=0.2] (-0.75,-1.45) rectangle (-0.25,-0.05);
    \draw[fill=white] (-0.75,1.05) rectangle (-0.25,2.45);
    \fill[red,opacity=0.2] (-0.75,1.05) rectangle (-0.25,2.45);
    \draw[rounded corners, fill=white] 
    (0.25,-1.75) rectangle (2.25,2.75);
    \fill[rounded corners, blue, opacity=0.1]
    (0.25,-1.75) rectangle (2.25,2.75);
    \node at (1.25,0.5) {$R_2$};
    \end{tikzpicture} 
\end{align*}
    \caption{\emph{Best-case scenario for two errors:} One of the errors, say $E_2$, commutes with the relevant circuit part $R_2$. Reordering allows us to treat the two errors as a single effective error
 $\tilde{E}=E_1 \circ E_2$.
In addition, $E_1$ and $E_2$ affect disjoint bit collections (independence) and $\tilde{E}$ factorizes nicely into two disjoint components:
 $\mathrm{Pr}\big[\tilde{R}(\vec{x}) \neq R(\vec{x})\big] = \mathrm{Pr} \big[\tilde{E}(\vec{x}) \neq \vec{x} \big]\geq 1- (1-2^{-(k-1)})^2$ 
 (quadratic improvement). 
    }
    \label{fig:best-case}
\end{figure}
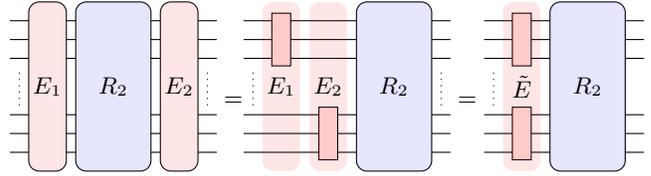

Let us first discuss $l=2$ errors of size $k$. An extension to multiple errors ($l \geq 3$) and different sizes will be straightforward. 
Fig.~\ref{fig:multiple-errors} provides valuable guidance for potential best-case behavior. Suppose that one of the errors, say $E_2$, can be pulled through the central circuit part $R_2$ without affecting it: $E_2 \circ R_2 = R_2 \circ E_2$. 
If circuit and error commute in such a fashion, we can group both errors into a single layer and have effectively reduced the problem to the single-error case which we already understand:
\begin{equation*}
\tilde{R}=R_3 \circ E_2 \circ R_2 \circ E_1 \circ R_1=
(R_3 \circ R_2) \circ (E_2 \circ E_1) \circ R_1.
\end{equation*}
The only remaining question is: what is the probability of failing to detect the cumulative error $E_2 \circ E_1$ with a single random input? This failure probability is smallest if the two errors are \emph{independent} in the sense that they act on disjoint sets of $k$ bits each.
A uniformly random input $\vec{x}\in \mathrm{Unif}(\left\{0,1\right\}^n)$ then ensures that the failure probability factorizes:
\[\resizebox{0.99\linewidth}{!}{$
\mathrm{Pr}\left[ (E_2 \circ E_1) (\vec{x}) = \vec{x}\right]=\prod_{i=1}^2 \mathrm{Pr} \left[ E_i (\vec{x})=\vec{x}\right] 
\leq  \left(1-2^{-(k-1)}\right)^2$}.
\]
This argument readily extends to multiple errors ($l \geq 3$). Taking the complement ensures
\begin{align}
\mathrm{Pr}\left[ \tilde{R}(\vec{x})\neq R(\vec{x})\right]=& 1-
\mathrm{Pr} \left[ E_l \circ \cdots \circ E_1 (\vec{x}) = \vec{x} \right] \nonumber \\
\geq &1- \left(1- 2^{-(k-1)} \right)^l
, \label{eq:best-case-aux}
\end{align}
provided that all $l$ errors commute  with the circuit (first equality) and act on different subsets of $k$ bits each (second inequality).
Rel.~\eqref{eq:best-case-aux} highlights that the probability of (best case) error detection increases substantially with the number of errors $l$. Intuitively, this makes sense: more errors should be easier to detect. 
This insight has implications for the number $N$ of random inputs that are required to detect $l$ best-case errors of size $k$ each. 
To pinpoint them, it is instructive to view a single simulation run as a biased coin toss:
we detect a discrepancy with probability
$p =\mathrm{Pr}\big[\tilde{R}_2 (\vec{x})\neq R_2 (\vec{x})\big]$ (``heads'') and fail to detect it with probability $1-p=\mathrm{Pr}\big[ \tilde{R}_2 (\vec{x}) = R_2 (\vec{x}) \big]$ (``tails''). 
When attempting to detect a discrepancy, we input new randomly generated inputs until we find a mismatch. This is equivalent to tossing the biased coin until ``heads'' appears. 
The expected number of required coin tosses to achieve this goal is $1/p$ (geometric distribution). Together with Rel.~\eqref{eq:best-case-aux}, this analogy allows us to conclude that we expect to require
\begin{align}
N_{\mathrm{expect}}^{(\downarrow)}
\leq &\frac{1}{1- \left(1-2^{-(k-1)}\right)^l}
& 
\text{ (best case)} 
\label{eq:best-case-expectation}
\end{align}
random inputs to detect $l$ commuting and independent errors of size $k$ each. 
This bound is sharp. It holds with equality if each of the $l$ errors is a worst-case error of size $k$, e.g.\ a $(k-1)$-fold controlled \textsc{NOT} gate.

We conclude this section with a simplified interpretation of Rel~\eqref{eq:best-case-expectation}.
For small $l$ (in comparison to $2^{(k-1)}$), the claim is comparable to $ N_{\mathrm{expect}}^{(\downarrow)} \approx 2^{k-1}/l$, which can also be observed in Fig.~\ref{fig:numerics}: the slopes of the solid lines match this estimate rather well whenever the number of errors $l$ is small compared to $2^{(k-1)}$.
Under best-case assumptions, detecting $l$ size $k$-errors is $l$-times easier than detecting a single error of the same size.

\subsection{Worst-case: anti-commuting errors and masking}
\label{sub:worst-case}

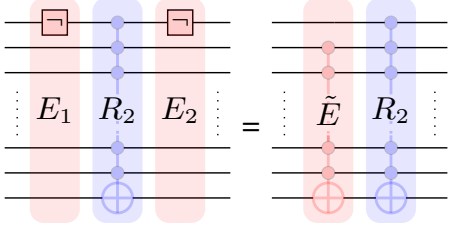
\begin{figure}
    \centering
    \resizebox{0.7\linewidth}{!}{
\begin{tikzpicture}[baseline,scale=0.5]
    \foreach \y in {-1.25,-0.75,-0.25,1.25,1.75,2.25}
    {\draw (-2.25,\y) -- (-0.125,\y);
    \draw (0.125,\y) -- (2.25,\y);
    };
    \draw[dotted] (-2,0) -- (-2,1);
    \draw[dotted] (2,0) -- (2,1);
    \foreach \y in {-0.75,-0.25,1.25,1.75,2.25}
    {\draw[fill=blue,opacity=0.2] (0,\y) circle(0.125);
    \draw[thick,blue,opacity=0.2] (0,\y-0.125) -- (0,\y-0.375);
    };
    \draw[thick,blue,opacity=0.2,dotted] (0,0.875) -- (0,0);
    \draw[thick,blue,opacity=0.2] (0,-0.125) -- (0,0);
    \draw[thick,blue,opacity=0.2] (0,-1.25) circle (0.3);
    \draw[thick,blue,opacity=0.2] (0,-0.75-0.375) -- (0,-1.55);
    \draw[thick,white] (-0.3,-1.25) -- (0.3,-1.25);
    \draw[thick,blue,opacity=0.2] (-0.3,-1.25) -- (0.3,-1.25);
    \fill[rounded corners, red,opacity=0.1] 
    (-1.75,-1.75) rectangle (-0.75,2.75);
    \node at (-1.25,0.5) {$E_1$};
    \draw[fill=white] (-1.5,2) rectangle (-1,2.5);
    \fill[red,opacity=0.2] (-1.5,2) rectangle (-1,2.5);
    \node at (-1.25,2.25) {\scriptsize $\neg$};
    \fill[rounded corners, blue, opacity=0.1]
    (-0.5,-1.75) rectangle (0.5,2.75);
    \node at (0,0.5) {$R_2$};
    \fill[rounded corners, red, opacity=0.1] 
    (0.75,-1.75) rectangle (1.75,2.75);
    \node at (1.25,0.5) {$E_2$};
    \draw[fill=white] (1,2) rectangle (1.5,2.5);
    \fill[red,opacity=0.2] (1,2) rectangle (1.5,2.5);
    \node at (1.25,2.25) {\scriptsize $\neg$};
    \end{tikzpicture} 
     = 
    \begin{tikzpicture}[baseline,scale=0.5]
    \foreach \y in {-1.25,-0.75,-0.25,1.25,1.75,2.25}
    {\draw (-1.75,\y) -- (-0.75,\y);
    \draw (-0.5,\y) -- (0.5,\y);
    \draw (0.75,\y) -- (1.75,\y);
    };
    \draw (-0.75,2.25) -- (-0.5,2.25);
    \draw[dotted] (-1.5,0) -- (-1.5,1);
    \draw[dotted] (1.5,0) -- (1.5,1);
    \draw[fill=blue,opacity=0.2] (0.625,2.25) circle(0.125);
    \draw[thick,blue,opacity=0.2] (0.625,2.25-0.125) -- (0.625,2.25-0.375);
    \draw[thick,blue,opacity=0.2] (0.625,-0.125) -- (0.625,0);
    \draw[thick,blue,opacity=0.2,dotted] (0.625,0.875) -- (0.625,0);
    \draw[thick,red,opacity=0.2] (-0.625,-0.125) -- (-0.625,0);
    \draw[thick,red,opacity=0.2,dotted] (-0.625,0.875) -- (-0.625,0);
    \foreach \y in {-0.75,-0.25,1.25,1.75}
    {\draw[fill=blue,opacity=0.2] (0.625,\y) circle(0.125);
    \draw[thick,blue,opacity=0.2] (0.625,\y-0.125) -- (0.625,\y-0.375);
    \draw[fill=red,opacity=0.2] (-0.625,\y) circle(0.125);
    \draw[thick,red,opacity=0.2] (-0.625,\y-0.125) -- (-0.625,\y-0.375);
    };
    draw[thick,blue,opacity=0.2] (0.625,-0.125) -- (0.625,0);
    \draw[thick,blue,opacity=0.2] (0.626,-1.25) circle (0.3);
    \draw[thick,blue,opacity=0.2] (0.625,-0.75-0.375) -- (0.625,-1.55);
    \draw[thick,white] (-0.3+0.625,-1.25) -- (0.3+0.625,-1.25);
    \draw[thick,blue,opacity=0.2] (-0.3+0.625,-1.25) -- (0.3+0.625,-1.25);
    \draw[thick,red,opacity=0.2] (-0.625,-0.125) -- (-0.625,0);
    \draw[thick,red,opacity=0.2] (-0.626,-1.25) circle (0.3);
    \draw[thick,red,opacity=0.2] (-0.625,-0.75-0.375) -- (-0.625,-1.55);
    \draw[thick,white] (-0.3-0.625,-1.25) -- (0.3-0.625,-1.25);
    \draw[thick,red,opacity=0.2] (-0.3-0.625,-1.25) -- (0.3-0.625,-1.25);
    \fill[rounded corners, blue, opacity=0.1]
    (0.125,-1.75) rectangle (1.125,2.75);
    \node at (0.625,0.5) {$R_2$};
    \fill[rounded corners, red, opacity=0.1]
    (-1.125,-1.75) rectangle (-0.125,2.75);
    \node at (-0.625,0.5) {$\tilde{E}$};
    \end{tikzpicture}
    }
    \caption{\emph{Worst-case scenario for two errors:} 
    Two bit-flip errors ($k=1$) affect one control line of a $(n-1)$-fold controlled \textsc{NOT}-gate. 
    These errors do not commute with the relevant circuit part $R_2$. Quite the opposite: two errors with size $k=1$ produce an effective error $\tilde{E}$ of size $k=(n-1)$. To make matters even worse, such a $(n-2)$-fold controlled \textsc{NOT} error is extremely difficult to detect:
    $\mathrm{Pr}\big[\tilde{R}(\vec{x})\neq R(\vec{x})\big]=\mathrm{Pr}\big[\tilde{E}(\vec{x})\neq \vec{x}\big]=4/2^{n}$ (masking).
    }
    \label{fig:worst-case}
\end{figure}

We expect that worst case error accumulation should occur when errors and relevant circuit portion do not commute at all (``anti-commutation''). If this is the case, the probability of detecting errors can become exponentially small in the total number of bits. We illustrate this by means of an example that is illustrated in Fig.~\ref{fig:worst-case}: $E_1$ and $E_2$ are bit-flip errors ($k=1$) that affect the first bit while $R_2: \left\{0,1\right\}^n \to \left\{0,1\right\}^n$ is a $(n-1)$-fold controlled \textsc{NOT}-gate. 
It is easy to check that
\begin{equation*}
E_2 \circ R_2 \circ E_1 = \tilde{E} \circ R_2,
\end{equation*}
where $\tilde{E}$ is a $(n-2)$-fold controlled \textsc{NOT} gate that acts on all bits, except the very first one ($k=n-1$).
This is a single worst-case error of almost maximal size. Proposition~\ref{prop:no-masking} and Lemma~\ref{lem:error-size} assert
\begin{equation*}
\mathrm{Pr} \big[ \tilde{R}(\vec{x}) \neq R(\vec{x})\big]= \mathrm{Pr} \left[\tilde{E}(\vec{x}) \neq \vec{x} \right]= \frac{4}{2^n}.
\end{equation*}
This success probability is exponentially small in the total number of bits and we expect to require a total of
\begin{align}
N_{\mathrm{expect}}^{(\uparrow)} \geq & 2^{(n-2)}
& \text{(worst case)} \label{eq:worst-case-expectation}
\end{align}
random inputs in order to detect the discrepancy. Even worse error accumulation effects can occur for more errors ($l \geq 3$) and/or larger error sizes ($k \geq 2$). But already 
Rel.~\eqref{eq:worst-case-expectation} is almost as bad as it can be. It is only a factor of two away from $2^{n-1}$---the absolute worst case for distinguishing \emph{any} pair of reversible circuits, see Lemma~\ref{lem:reversible} (iii).

\subsection{Empirical studies}

The multiple-error case is intricate by comparison, because the interplay between error (locations) and underlying circuit geometry starts to matter. We have seen that this leads to strikingly different best- (commuting errors, Sub.~\ref{sub:best-case}) and worst-case (anticommuting errors, Sub.~\ref{sub:worst-case}) behavior. 
Concrete problem instances fall into the wide range between these extreme cases. 
In this section, we employ numerics to delineate \emph{typical} behavior. 

We study the effect of size-$k$ errors in reversible circuits with~$n$ lines.
For a given number of lines~$n$, we construct random reversible circuits with~$g\approx\mathcal{O}(n^2)$ arbitrary multi-controlled NOT gates.
When injecting errors of size $k$, we always consider $(k-1)$-fold controlled NOT gates which represent the worst case behavior, as discussed in Section~\ref{sub:size}.
Without loss, we assume that these errors are 
geometrically local,~i.e.,~they only affect neighbouring lines. 
All experiments were repeated $\num{10000}$ times with different random seeds in order to ensure adequate statistical uniformity.

\begin{figure}[htbp]
\centering
\includegraphics[width=0.99\linewidth]{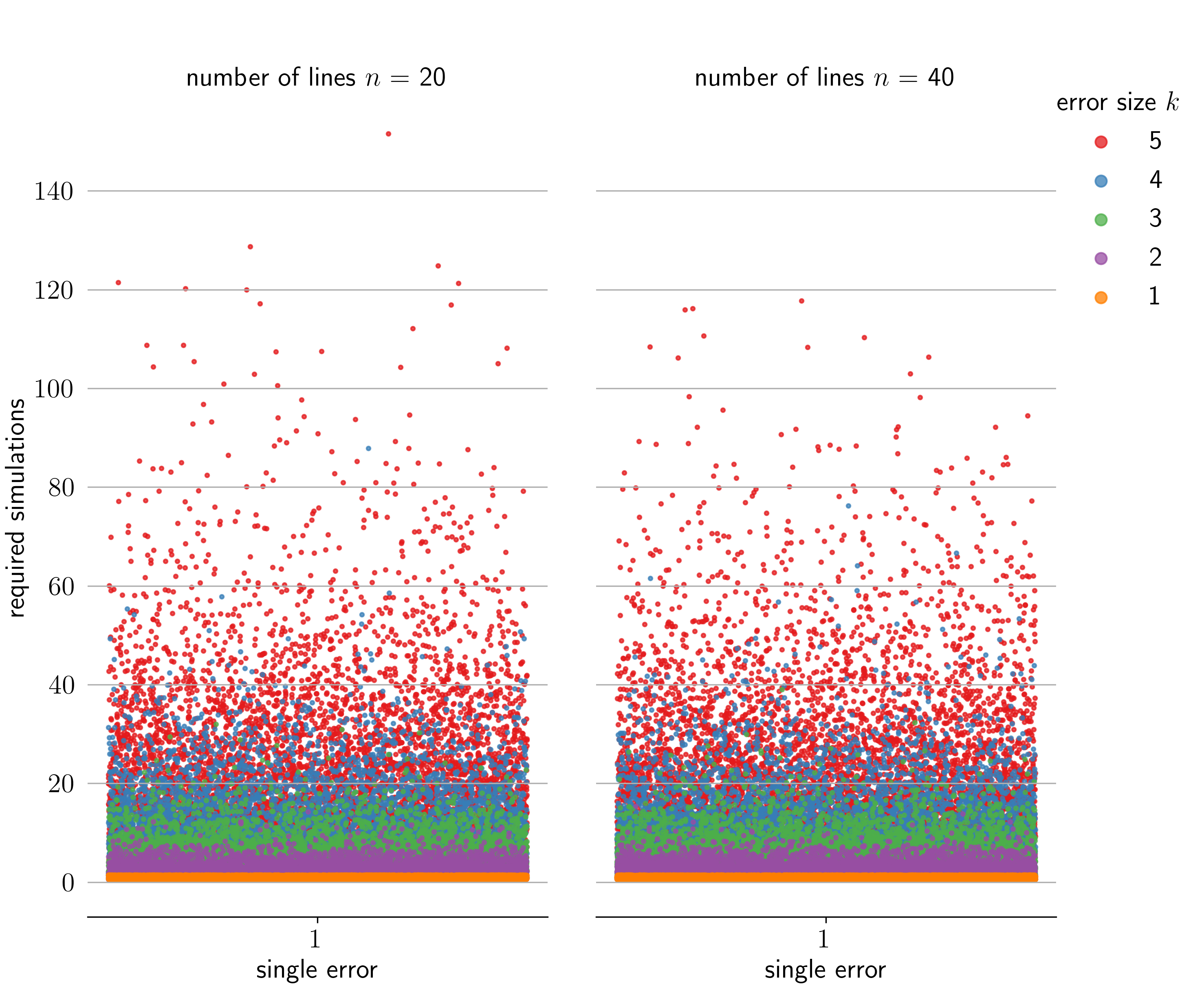}
\caption{\emph{Confirmation of theoretical results}: Scatter-plot of required simulations ($y$-axis) for detecting a single error of size $k$ in a circuit with $n=20$ (left plot) and $n=40$ (right plot) lines. Different colors denote varying values of \mbox{$k\in\{1, 2, 3, 4, 5\}$} . This experimentally confirms that the distribution of simulations does not depend on the number of lines and that the number of required simulations grows exponentially with the size of the error.}
\label{fig:theory-backup}
\end{figure}

First and foremost, we confirm interesting aspects of the theory developed in Sec.~\ref{sec:theory}.
To this end, we considered the injection of a single size-$k$-error and count
the required number of simulations for detecting this error. The results are depicted in Fig.~\ref{fig:theory-backup}.
In contrast to classical intuition,
the probability of detecting a single error of size $k$ is (1) completely independent of the circuit under consideration, and (2) diminishes exponentially in the error size $k$, i.e., the smaller the error, the greater its impact.
This is in excellent agreement with Theorem~\ref{thm:main-restatement}.
On average,
the required simulations exactly follow the predicted $2^{k-1}$ trajectory with no apparent variation. Additionally, the distributions of results is the same when simulating the circuits $R=R_2\circ R_1$ and $\tilde{R}=R_2\circ E\circ R_1$ as compared to only simulating the error $E$ itself.

The next set of numerical experiments pilots us in more interesting territory. Namely, the multiple-error case.
We have already teased the results in the introduction and summarized them in Fig.~\ref{fig:numerics}. 
The averaged number of inputs highlights an excellent agreement between the observed behavior and the best-case scenario discussed in Section~\ref{sub:best-case}. 
The deviation from this optimum for higher numbers of errors can be explained by accumulation affects of errors not acting independently (see Section~\ref{sub:worst-case}).

\begin{figure}[htbp]
\centering
\includegraphics[width=\linewidth]{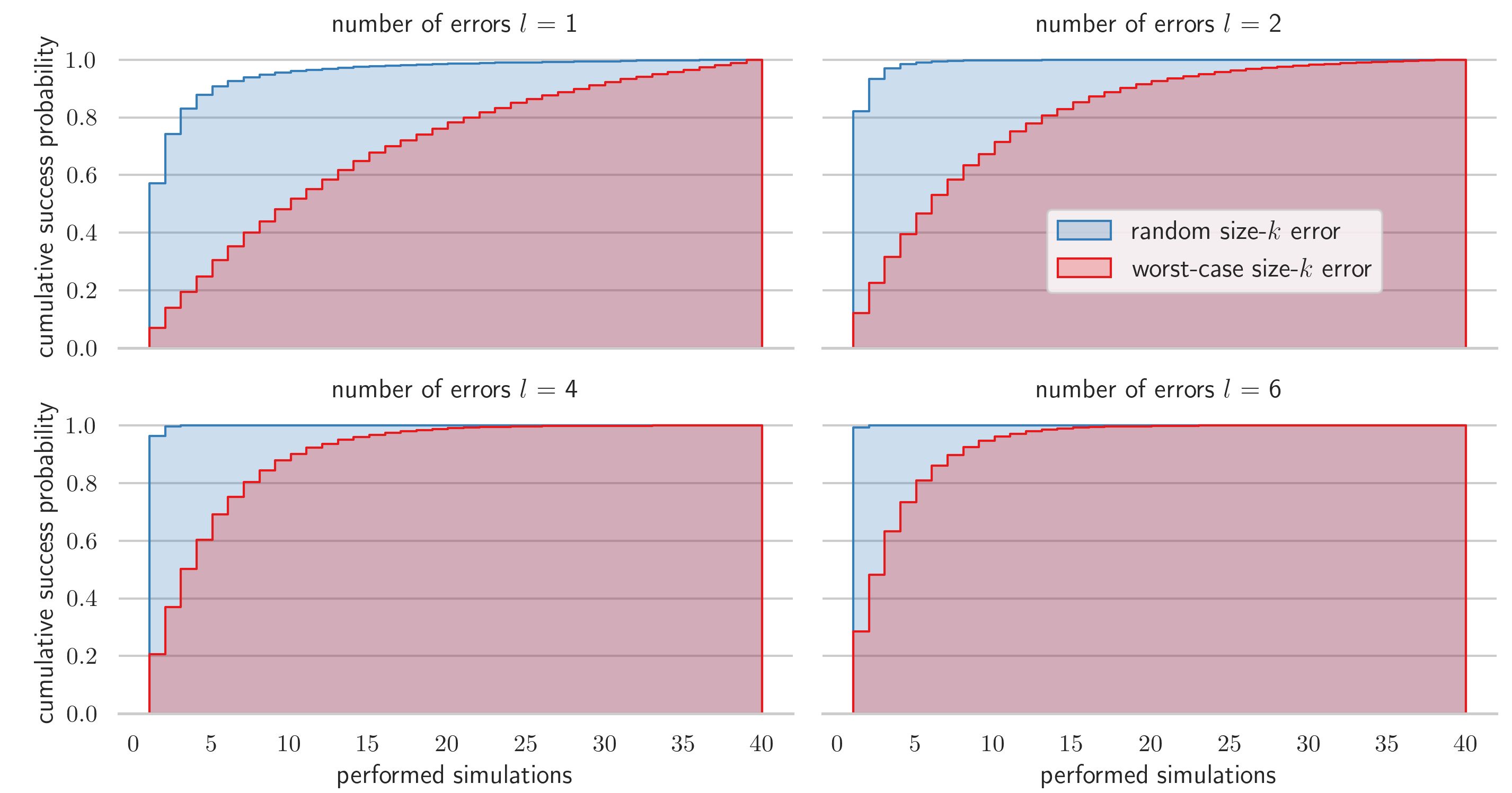}
\caption{\emph{Comparison of worst-case and average-case errors }: performed simulations ($x$-axis) vs.\ cumulative distribution function (cdf) for detecting $l=1,2,4,6$ errors of size $k=5$ (\mbox{$y$-axis}). The red curve corresponds to injecting worst-case errors, 
while the blue
curve delineates the cdf for detecting \emph{randomly generated} errors of the same size. 
This goes to show, that 
average-case errors require far less simulations than worst-case ones.
}
\label{fig:average-case}
\end{figure}

Last but not least, we emphasize that---up to this point---theoretical and empirical results have been contingent on a worst-case assumption: each injected size-$k$ error is a $(k-1)$-fold controlled \text{NOT}-gate. 
In a final series of evaluations, we analyzed the success probability after conducting a certain number of simulations when choosing errors \emph{at random}. More precisely, each size-$k$ error is a 
randomly selected gate sequence with the additional constraint that none of the $k$ relevant lines remain unaffected (such a scenario would produce an error of size (at most) $(k-1)$). 
We expect that this error model captures typical behavior in a more accurate fashion.
The results are shown in Fig.~\ref{fig:average-case}
and highlight a considerable discrepancy between random (blue) and worst-case (red) errors. 
This is not at all surprising. Random errors of size $k$ tend to factorize into several independent contributions and the probabilities of detecting them with random inputs factorizes accordingly, see Sub.~\ref{sub:best-case}. 
Such factorizations lead  
to an increased error detection probability 
within (very) few simulation runs.

\section{Conclusion}
In this work, we have shown the impact of the reversible circuit paradigm on the probability of detecting errors in circuits.
Our rigorous analysis shows, that, as opposed to classical/irreversible circuits, reversible circuits can never mask single errors and, that the probability of detecting a single error only depends on the error's size and not at all on the surrounding circuit.
Empirical evaluations have shown that, in case of multiple errors, the detection probability is very close to the theoretical best case.
Finally, we have observed that, in case the assumption of worst-case errors is dropped, the probability of detecting these errors is increased even more.

\section*{Acknowledgments} 
The authors want to thank J. K\"ung for inspiring discussions throughout the early stages of this project.

This work has partially been supported by the LIT Secure and Correct Systems Lab funded by the State of Upper Austria as well as by BMK, BMDW, and the State of Upper Austria in the frame of the COMET Programme managed by FFG.
\printbibliography

\end{document}